\documentclass{jgaa-art}
\usepackage[utf8]{inputenc}
\usepackage{amsmath,amssymb}
\usepackage{xspace}
\usepackage{graphicx}
\DeclareGraphicsExtensions{.pdf,.png,.eps}
\graphicspath{{figs/}} 

\usepackage{enumitem}

\usepackage{mathtools}
\DeclarePairedDelimiter\set{\{}{\}}
\DeclarePairedDelimiter\abs{\lvert}{\rvert}
\def\to{\ensuremath{\rightarrow}}
\def\sL{\ensuremath{\nwarrow}} 
\def\sR{\ensuremath{\nearrow}}
\def\Oh{\ensuremath{\mathcal{O}}}
\def\skel{\ensuremath{\mathit{skel}}} 
\newcommand{\cF}{{\mathcal{F}}}
\newcommand{\cT}{{\mathcal{T}}}
\DeclareMathOperator{\n}{\mathrm{n}}

\newtheorem{corollary}[theorem]{Corollary}
\newtheorem{observation}[theorem]{Observation}

\usepackage{url}
\usepackage{cite}
\usepackage[capitalise,noabbrev,nameinlink]{cleveref}
\newcommand{\etal}{{et~al.}}
\begin{document} 

\doi{10.7155/jgaa.00587}
% --------------------------------------------------------------------
\Issue{26}{1}{171}{198}{2022} % volume, number, start page, end page, year
% --------------------------------------------------------------------
\HeadingAuthor{Klawitter and Mchedlidze} 
\HeadingTitle{Upward planar drawings with two slopes}
% --------------------------------------------------------------------
\title{Upward planar drawings with two slopes}
% --------------------------------------------------------------------
\authorOrcid[first]{Jonathan Klawitter}{jo.klawitter@gmail.com}{0000-0001-8917-5269}
\affiliation[first]{University of Würzburg, Germany}
\authorOrcid[second]{Tamara Mchedlidze}{t.mtsentlintze@uu.nl}{0000-0002-1545-5580}
\affiliation[second]{Utrecht University, The Netherlands}
% --------------------------------------------------------------------
\submitted{November 2021}%
\reviewed{March 2021}%
\revised{April 2021}%
\accepted{May 2021}%
\final{May 2021}%
\published{June 2021}%
\type{Regular paper}%
\editor{G. Liotta}%
% --------------------------------------------------------------------

\maketitle

%\pdfbookmark[1]{Abstract}{Abstract}
\begin{abstract}
In an upward planar 2-slope drawing of a digraph, 
edges are drawn as straight-line segments in the upward direction 
without crossings using only two different slopes. 
We investigate whether a given upward planar digraph admits such a drawing 
and, if so, how to construct it.
For the fixed embedding scenario, we give a simple characterisation 
and a linear-time construction by adopting algorithms from orthogonal drawings.
For the variable embedding scenario, we describe a linear-time algorithm for single-source digraphs, 
a quartic-time algorithm for series-parallel digraphs, 
and a fixed-parameter tractable algorithm for general digraphs. 
For the latter two classes, we make use of SPQR-trees and the notion of upward spirality.
As an application of this drawing style, 
we show how to draw an upward planar phylogenetic network with two slopes 
such that all leaves lie on a horizontal line. 
\end{abstract} 

\section{Introduction} \label{sec:introduction} % -----------------------------
When we visualize directed graphs (digraphs for short) that model hierarchical relations
with node-link diagrams,
we traditionally turn edge directions into geometric directions by letting each edge point upward.
Aiming for visual clarity, we would like such an upward drawing to be planar, 
that is, no two edges should cross~\cite{DETT99}.
If this is possible, the resulting drawing is called \emph{upward planar drawing}; see \cref{fig:introExample}~(a).
Interestingly, as Di Battista and Tamassia~\cite{DT88} have shown,
every upward planar drawing can be turned into one where each edge is drawn with a single line segment;
such a \emph{straight-line drawing} may however require an exponentially large drawing area~\cite{DTT92}.

\begin{figure}[tb]
  \centering
  \includegraphics{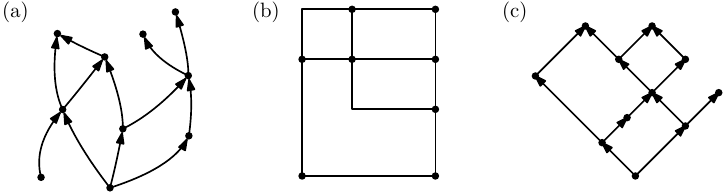}
  \caption{(a) An upward planar drawing; (b) an orthogonal drawing; (c) an upward planar 2-slope drawing.}
  \label{fig:introExample}
\end{figure}

Another important class of drawings are \emph{(planar) orthogonal drawings},
where edges are drawn as sequences of horizontal and vertical line segments~\cite{DETT99}; see \cref{fig:introExample}~(b).
This drawing style is commonly used for schematic drawings such as VLSI circuit design and UML diagrams.
Schematic drawings that allow more than two slopes for edge segments include hexalinear and octilinear drawings,
which find application in metro maps~\cite{NW11}.
In general, the use of only few geometric primitives (such as different slopes)
in a graph drawing facilitates a low visual complexity; a common quality measure for drawings~\cite{Sch15}. 

In recent years, the interest in upward planar drawings that use only few different slopes has grown.
For example, among other results, Bekos \etal~\cite{BDDLM18} showed
that every so-called bitonic $st$-graph with maximum degree~$\Delta$
admits an upward planar drawing where every edge has at most one bend
and the edges segments use only~$\Delta$ distinct slopes. 
Di Giacomo \etal~\cite{DLM20} provided complementary results
by proving that also every series-parallel digraph admits a 1-bend upward planar drawing
on~$\Delta$ distinct slopes; their drawings also have optimal angular resolution.
Brückner \etal~\cite{BKM19} considered level-planar drawings,
that is, upward planar drawings 
where each vertex is drawn on a predefined integer y-coordinate (its level),
with a fixed slope set.
In this paper, we continue this recent trend. 
In particular, we study bendless upward planar drawings that use only two different slopes. 
An example of such a drawing is shown in \cref{fig:introExample}~(c).
Some of our results can be extended to 1-bend upward planar drawings.
We now define these drawing concepts more precisely and list related work.

\paragraph{Upward planarity.}
An \emph{upward planar drawing} of a directed graph~$G$ is a planar drawing of $G$
where every edge $(u, v)$ (i.e., an edge directed from $u$ to $v$) 
is drawn as a monotonic upward curve from $u$ to $v$.
A digraph is called \emph{upward planar} if it admits an upward planar drawing. 
Two upward planar drawings of the same digraph are \emph{topologically equivalent} if the left-to-right orderings 
of the incoming and outgoing edges around each vertex coincide in the two drawings. 
An \emph{upward planar embedding} is an equivalence class of upward planar drawings.  
An upward planar digraph is called \emph{upward plane} if it is equipped with an upward planar embedding. 

A necessary though not sufficient condition for upward planarity is acyclicity~\cite{BDLM94}.
Moreover, Garg and Tamassia~\cite{GT01} showed that testing upward planarity is NP-complete for general digraphs.
While the digraphs used in their reduction contain vertices with in- and outdegree higher than two,
such vertices can be split into mulitple vertices of maximum in- and outdegree at most two
without losing any of the properties required in their proofs. 
It is thus also NP-complete to test upward planarity for digraphs 
with in- and outdegree at most two.
On the positive side, there exist several fixed-parameter tractable algorithms for general digraphs~\cite{Cha04,HL06,DGL10} 
and polynomial time algorithms for single source digraphs~\cite{BDMT98}, series-parallel digraphs~\cite{DGL10}, 
outerplanar digraphs~\cite{Pap95}, and triconnected digraphs~\cite{BDLM94}. 
Moreover, upward planarity can be decided in polynomial time if the embedding is specified~\cite{BDLM94}.

\paragraph{$\ell$-bend $k$-slope drawings.}
In an \emph{$\ell$-bend drawing} of a graph $G$ each edge is drawn with at most $\ell+1$ line segments; 
equivalently, each edge has at most $\ell$ bends.
An \emph{$\ell$-bend $k$-slope drawing} of $G$ is an $\ell$-bend drawing of $G$ 
where every edge segment has one of at most $k$ distinct slopes.
From now on and if not further specified, 
we refer to bendless (or 0-bend) $k$-slope drawings simply as \emph{$k$-slope drawings}.

Note that orthogonal drawings are $2$-slope drawings without a bound on the number of bends.
Tamassia~\cite{Tam87} showed that a planar orthogonal drawing with minimum total number of bends 
of a plane graph on $n$ vertices can be computed in~$\Oh(n^2 \log n)$ time. 
Rhaman \etal~\cite{RNN03} gave necessary and sufficient conditions for a subcubic plane graph
to admit a bendless orthogonal drawing. 
For drawings of cubic graphs in 3D, 
Eppstein~\cite{Epp13} considered bendless orthogonal graph drawings where
two vertices are adjacent if and only if two of their coordinates are equal. 

Given a graph $G$, the minimum number $k$ of slopes needed for $G$ 
to admit a $k$-slope drawing is called the \emph{slope number} of $G$~\cite{WC94}.
In the planar setting, this is the \emph{planar slope number} of $G$.
Both these numbers have been studied extensively.
For example, Pach and Pálvölgyi~\cite{PP06} showed that the slope number of graphs with maximum degree~5
can be arbitrarily large. Further results, include bounds on slope numbers of graph classes
such as trees, 2-trees, planar 3-trees, outerplanar graphs~\cite{DESW07,LLMN13,JJKLTV13,KMW14}, 
cubic graphs~\cite{MP12}, and subcubic graphs~\cite{DLM18,KMSS18}.
Determining the planar slope number is hard in the existential theory of the reals~\cite{Hof17}.

\paragraph{Upward planar 2-slope drawings.}
The focus of this paper lies on (bendless) upward planar $2$-slope drawings. 
We consider only the slope set $\set{-\pi/4,\pi/4}$ and denote it by $\set{\sL, \sR}$,
since an upward planar 2-slope drawing on any two slopes can be morphed into an upward planar 2-slope drawing 
with the slopes $\sL$ and $\sR$ -- imagine this as (un)skewing a partial grid; see \cref{fig:skewing}.
Note that a natural lower bound on the upward planar slope number of a graph is given by its maximum in- and outdegree.
Hence, we assume that the graphs considered in this paper have maximum in- and outdegree at most two.

\begin{figure}[htb]
  \centering
  \includegraphics{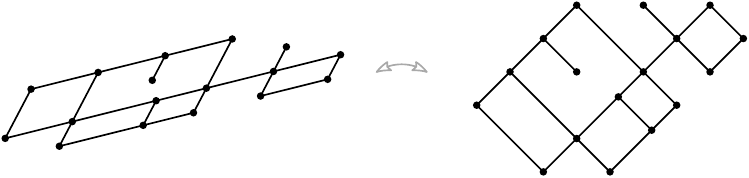} 
  \caption{Two upward planar 2-slope drawings of the same graph on different slope sets -- edge directions are given implicitly. 
  Using an affine transformation one can transform a drawing on any size-two slope set into one on $\set{\sL, \sR}$.}
  \label{fig:skewing}
\end{figure}

Bachmaier \etal~\cite{BBBHMU09}, Brunner and Matzeder~\cite{BM11}, and Bachmaier and Matzeder~\cite{BM13}
studied straight-line drawings of ordered and unordered rooted trees on orthogonal grids with $k$ directions for $k \in \set{4, 6, 8}$. 
Some of their drawing styles are also upward planar.
A classical result of Crescenzi \etal~\cite{CDP92} shows
that any binary tree with $n$ vertices admits an upward planar 2-slope drawing in $\Oh(n \log n)$ area.  
Concerning more complex graphs, 
upward planar drawings with few slopes for lattices have been studied by Czyzowicz \etal~\cite{CPR90} and Czyzowicz~\cite{Czy91}.
As mentioned above, Bekos \etal~\cite{BDDLM18} and Di Giacomo \etal~\cite{DLM20} 
considered such drawings for $st$-graph and series-parallel graphs
but also allowed bends.  
In a companion paper to the current one, 
Klawitter and Zink~\cite{KZ21} study upward planar $k$-slope drawings for $k \geq 3$
and among other results show that it is NP-hard to decide whether an outerplanar digraph
admits an upward planar 3-slope drawing.

\paragraph{Phylogenetic networks.}
Our interest in upward planar 2-slope drawings 
also stems from the problem of visualizing phylogenetic networks.
\emph{Phylogenetic trees} and \emph{networks} are used to model 
the evolutionary history of a set of taxa like species, genes, or languages~\cite{HRS10,Dun14,Ste16}.
The precise definition of phylogenetic networks and their drawing conventions may vary widely depending on the particular use case.
For instance, vertices may have timestamps that should be represented in the drawings or
leaves may be required to be placed on the same height. 
In combinatorial phylogenetics the following definition is commonly used~\cite{HRS10}:
A \emph{phylogenetic network} is a rooted digraph where the leaves are labelled bijectively with a set of taxa. 
Inner vertices are either \emph{tree vertices} that have indegree one and outdegree two 
or \emph{reticulations} that have indegree two and outdegree one; see \cref{fig:phynets}~(a).
A network without reticulations is a \emph{phylogenetic tree}; see \cref{fig:phynets}~(b,c).

\begin{figure}[htb]
  \centering
  \includegraphics[page=1]{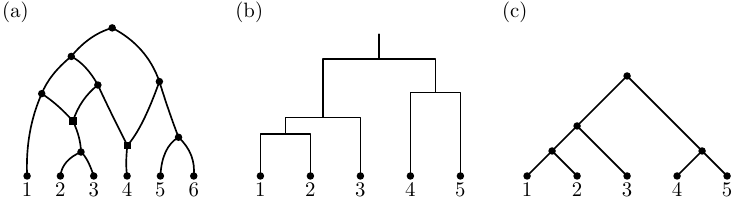}
  \caption{(a) A phylogenetic network with two reticulations; 
  (b) a phylogenetic tree drawn as rectangular cladogram and (c) upward planar with two slopes.}
  \label{fig:phynets}
\end{figure}

There exist different drawing styles for phylogenetic trees such as rectangular or circular cladograms~\cite{BBS05,Hus09}.
If the focus is on the topology of the tree (and thus the taxonomy), 
a common drawing style is upward planar with 2-slope and all leaves aligned on a line.
Theoretical work to adapt classical drawing styles from phylogenetic trees to phylogenetic networks 
has been carried out by Huson \etal~\cite{KH08,Hus09,HRS10}.
A different approach has been taken by Tollis and Kakoulis~\cite{TK16}, 
who propose a drawing style similar to treemaps for a special class of phylogenetic networks.
There also exist several software tools to draw phylogenetic networks~\cite{HB05,HS12,BDM12,Vau17,SVBDAS21}.
Here we are interested in drawing upward planar phylogenetic networks with two slopes and the additional constraint
that all leaves lie on a horizontal line; see \cref{fig:phynets}~(c).

\paragraph{Contribution.}
In this paper we investigate the following decision problem. 
Given a digraph $G$ of in- and outdegree at most two, decide whether it admits an upward planar $2$-slope drawing.
We distinguish the fixed and variable embedding scenario, that is, whether $G$ is already equipped with an upward planar embedding or not.
In the former case, 
we give a simple characterisation of when a drawing exists; see \cref{sec:plane}.
By making use of orthogonal drawing algorithms, we also show how to construct a drawing (if it exists) in linear time.
In addition, if no upward planar 2-slope drawing exists, we describe how to obtain an upward planar 1-bend 2-slope drawing of $G$ 
with minimum number of bends.

For the variable embedding scenario, we check whether graphs of different graph classes admit
an upward planar 2-slope drawing, based on the results of \cref{sec:plane}.
In \cref{sec:ss}, we show that for a single-source digraph $G$, checking whether an upward planar 2-slope drawing of $G$ exists 
can be done starting from any single upward planar embedding of $G$.
In the affirmative, a suitable upward planar embedding can be derived and a drawing constructed in linear time.
For series-parallel digraphs (\cref{sec:sp}) and general digraphs (\cref{sec:digraph,sec:nonbi}),
we derive a quartic-time and a fixed-parameter tractable algorithm, respectively.
These algorithms are based on Didimo et al.'s algorithms for upward planarity testing~\cite{DGL10}.

Lastly, we show how to compute 2-slope drawings of upward planar phylogenetic networks,
where all leaves lie on a horizontal line in linear time; see \cref{sec:phynet}.
We conclude with a short discussion and open problems.

\section{Preliminaries} \label{sec:prelim} 
Let $G$ be an upward plane digraph with maximum in- and outdegree two.
We assume, without loss of generality, that $G$ is connected and 
let $n$ denote the number of vertices of $G$ or the graph currently under consideration.
We use $(u, v)$ to denote an edge of $G$ that is directed from $u$ to $v$.
For two vertices $u, v \in V(G)$, we say that $u$ precedes $v$ if there is a directed path from $u$ to $v$.

If a vertex $v$ of $G$ has two incoming edges, 
then based on the left-to-right ordering of the edges around $v$, 
it is natural to talk about the \emph{left} and the \emph{right} incoming edge of $v$.
If an edge $e$ is the only incoming edge at $v$, we call $e$ the \emph{sole} incoming edge of $v$.
The same holds for outgoing edges; see \cref{fig:angles}~(a).
We say that a vertex $v$ has face $f$ to the \emph{left (right)}
if $f$ is the face left (resp. right) of $v$'s leftmost (resp. rightmost) incoming 
and leftmost (resp. rightmost) outgoing edge (if they exist).

A \emph{2-slope assignment} $\phi$ is a mapping from the edges of $G$ to the slopes $\sR$ and $\sL$,
that is, $\phi : E(G) \to \set{\sL, \sR}$.
We say $\phi$ is a \emph{consistent 2-slope assignment} if  
\begin{itemize}
  \item every left (right) incoming edge of a vertex is assigned the slope $\sR$ (resp. $\sL$), and
  \item every left (right) outgoing edge of a vertex is assigned the slope $\sL$ (resp. $\sR$). 
\end{itemize}
An edge $(u, v)$ that is the sole outgoing edge of $u$ and the sole incoming edge of $v$ may have either slope.
A digraph $G$ together with a consistent 2-slope assignment $\phi$ forms an \emph{upward planar 2-slope representation} $U_G = (G, \phi)$. 
To avoid cumbersome notation, we simply write 2-slope representation. 

Suppose $G$ contains an edge $e = (u, v)$ that is the left outgoing edge of $u$ and left incoming edge of $v$.
Let $f$ be the face to the right of $e$; see \cref{fig:angles}~(b).
We then call~$e$ a \emph{bad edge} with respect to $f$ since $e$ obstructs a consistent $2$-slope assignment.
Note, however, that $e$ is not a bad edge with respect to the face $f'$ left of $e$; see again \cref{fig:angles}~(b).
The same holds with ``left'' and ``right'' in reversed roles. 

\begin{figure}[htb]
  \centering
  \includegraphics{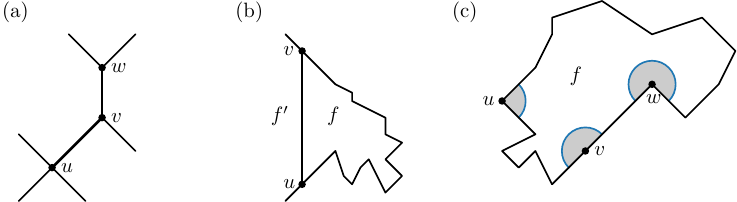}
  \caption{(a) The edge $(u, v)$ is the right outgoing edge of $u$ and left incoming edge of $v$, the edge $(v, w)$ is the sole outgoing edge of $v$;
  (b) the edge $(u, v)$ is a bad edge with respect to $f$;
  (c) $u$ is a left-switch spanning a small angle, $v$ spans a flat angle and is thus not a switch, and $w$ is a sink-switch spanning a large angle. }
  \label{fig:angles}
\end{figure}

Let $U_G$ be a $2$-slope representation of $G$.
Let $a = (e_1, v, e_2)$ be a triplet such that~$v$ is a vertex of the boundary
of a face $f$ and $e_1$, $e_2$ are incident edges of $v$ that are consecutive on the boundary of $f$ 
in counterclockwise direction. 
The following definitions follow Bertolazzi \etal~\cite{BDMT98} and Didimo \etal~\cite{DGL10} 
though are adjusted to also encapsulate geometric properties induced by~$U_G$.  
The triplet~$a$ is called an \emph{angle} of~$f$. We can categorise angles into three groups, 
namely, $a$ is
\begin{itemize}
  \item a \emph{large angle} if $e_1$ and $e_2$ span a 270$^{\circ}$ angle in $f$,
  \item a \emph{small angle} if $e_1$ and $e_2$ span a 90$^{\circ}$ angle in $f$, or
  \item a \emph{flat angle} if $e_1$ and $e_2$ span a 180$^{\circ}$ angle in $f$
\end{itemize}
with respect to the slopes assigned to~$e_1$ and~$e_2$.
A vertex~$v$ of~$G$ is called a \emph{local source} with respect to a face~$f$ 
if~$v$ has two outgoing edges on the boundary of~$f$. A \emph{local sink} is defined analogously.
Furthermore, we call~$v$ a \emph{switch} with respect to~$f$ if the slopes of~$e_1$ and~$e_2$ differ;
for example, every local source is a switch.
We further categorise switches by the angle they span and where they lie on the boundary of~$f$; see \cref{fig:angles}~(c).
A switch~$v$ is a \emph{large switch} if~$e_1$ and~$e_2$ span a large angle at~$f$ and a \emph{small switch} otherwise;
note that there can be no ``flat'' switches. 
We call~$v$ a \emph{source-switch} or \emph{sink-switch} if~$v$ is a local sink or local source, respectively.
Otherwise, if~$e_1$ and~$e_2$ have~$f$ to the right~(left), 
then~$v$ is a \emph{left-switch} (resp. \emph{right-switch}). 
Note that an inner face~$f$ of~$G$ contains exactly four small switches more than large switches and
that the outer face contains four large switches more than small switches.
An inner (the outer) face~$f$ is \emph{rectangular} if it contains exactly four small (resp. large) switches. 

Assume for now that~$G$ is biconnected. The following definitions are illustrated in \cref{fig:splitComponentFlip}.
A \emph{split pair}~$\set{u, v}$ of~$G$ is either a separation pair or a pair of adjacent vertices.
A \emph{split component} of~$G$ with respect to the split pair~$\set{u, v}$ is either an edge~$(u, v)$ (or~$(v, u)$)
or a maximal subgraph~$G'$ of~$G$ such that~$G'$ contains~$u$ and~$v$ and~$\set{u, v}$ is not a split pair of~$G'$.
Let~$G'$ be such a split component with respect to the split pair~$\set{u, v}$.
If~$G'$ is equipped with an upward planar embedding, then we define the \emph{flip} of~$G'$
as the change of the embedding of~$G'$ by reversing the edge ordering of every vertex of~$G'$.

\begin{figure}[htb]
  \centering
	\includegraphics{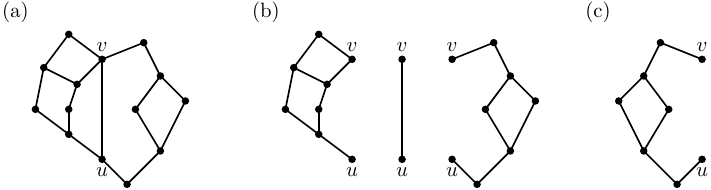}
  \caption{(a) A biconnected digraph with split pair $\set{u, v}$, 
  (b) which induces three split components; 
  (c) flip of the third split component.}
  \label{fig:splitComponentFlip}
\end{figure}

\section{Fixed embedding} \label{sec:plane} 
In this section we consider the problem of whether an upward plane digraph admits an upward planar 2-slope drawing 
under the given fixed embedding.
As noted above, a bad edge obstructs the existence of a consistent 2-slope assignment and thus of a 2-slope representation for $G$.
We show that the absence of any bad edges is not only necessary but also sufficient. 

Since upward planar 2-slope drawings are related to orthogonal drawings, we can make use of techniques used to construct them.
The classical algorithm by Tamassia~\cite{Tam87}, which constructs an orthogonal drawing of a plane graph with minimum number of bends, 
works in three steps; 
refer also to Di Battista \etal~\cite[Chapter 5]{DETT99}. 
It starts with a plane graph and constructs a so-called \emph{orthogonal representation}, a description of the shapes of the faces. 
The second step, called \emph{refinement}, subdivides each face into rectangles. 
Finally, the third step performs a so-called \emph{compaction} -- it assigns coordinates to the vertices with the goal to minimize the area of the drawing. 
The technique for constructing an orthogonal representation cannot be directly applied for the construction of an upward planar 2-slope drawing, 
as it does not preserve the upwardness of edges. 
However, assuming a 2-slope representation is already given, 
we can adopt the refinement algorithm by Tamassia~\cite{Tam87} and the compaction algorithm by Di Battista~\etal~\cite{DETT99} for our purposes.
In the following lemma we describe a modified version of Tamassia's algorithm that refines the faces of a 2-slope representation;
we explain how to obtain a 2-slope representation in \cref{clm:plane}.
For this, recall that a switch is a triplet $(e_1, v, e_2)$ consisting of a vertex and its two incident edges along a face in counterclockwise order
where $e_1$ and $e_2$ have different slopes.

\begin{lemma} \label[lemma]{clm:plane:refinement}
Let $G$ be an upward plane digraph on $n$ vertices with a 2-slope representation $U_G = (G, \phi)$.
Then, in $\Oh(n)$ time, $G$ can be refined into a digraph $\bar{G}$ that contains only rectangular faces
and such that $G$ is a topological minor of $\bar{G}$ respecting $\phi$. 
\end{lemma}
\begin{proof}
  If every face of $G$ is already rectangular, then we are done and $\bar{G} = G$.
  Assume that this is not the case and let $f$ be a non-rectangular inner face of $G$. 
  We describe how to refine $f$ into rectangular faces. 
  
  First, traverse the boundary of $f$ counterclockwise
  and store in each switch pointers to its preceding and subsequent switch.
  Next, starting at any switch, traverse the circular sequence of switches in counterclockwise order.
  Let $u$ be the first encountered large switch that is preceding a small switch $v$.
  Note that such $u$ exists since $f$ is not rectangular but contains at least four small switches.
  Without loss of generality, assume that $u$ is a sink-switch. (The cases when $u$ is a source-, left-, or right-switch work analogously.) 
  Let $w$ be the subsequent switch of $v$.  
  If $w$ is a large switch, then add a vertex~$x$ and the edges $(u, x)$ and $(w, x)$ with slopes $\sR$ and $\sL$, respectively; 
  see \cref{fig:augmentationStep}~(a).
  Otherwise, if $w$ is a small switch (and thus a right-switch), then subdivide the outgoing edge of $w$ with a new vertex $x$,
  and add the edge $(u, x)$ with slope $\sR$. Assign the slope~$\sL$ to the two edges resulting from the subdivision.
  This ensures that $\phi$ is respected by~$G$ as topological minor of $\bar{G}$; see \cref{fig:augmentationStep}~(b).
  In either of the two cases, the result is a rectangular face~$f_1$ 
  and a face $f_2$ with one less small and one less large switch than $f$. 
  Let $f_2$ now take the role of $f$ and store the preceding switch of $u$ and the subsequent switch of $w$ 
  as preceding and subsequent switches of $x$, respectively.
  If $x$ is a small switch, continue the traversal with the switch preceding~$x$ instead of with $x$.
  Therefore, if a large switch precedes $u$ in $f$, the process directly continues with a refinement step
  without having to potentially traverse the whole circular sequence first.
  Stop when only four switches are left and when $f$ is thus rectangular.
  Note that this process runs in linear time in terms of the size of $f$ and that it only adds as
  many new vertices as the number of large switches that $f$ contains.
      
  \begin{figure}[htb]
    \centering
    \includegraphics{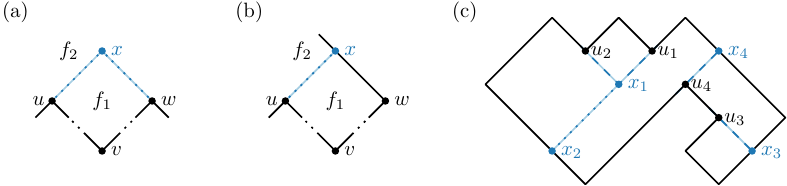}
    \caption{How to refine non-rectangular faces of $G$ to obtain $\bar{G}$ 
    when a large switch is followed (a) by a small and a large switch or (b) by two small switches.
    On the right, the large switch $u_4$ precedes the large switch $u_3$ and is thus processed after $u_3$.} 
    \label{fig:augmentationStep} 
  \end{figure}
  
  Repeat this procedure for every non-rectangular inner face of~$G$.
  If the outer face $f_0$ of $G$ is non-rectangular, 
  apply the analogous procedure with the difference that the goal is to remove all small switches
  such that $f_0$ only contains four large switches.
  Further note that here a large switch~$u$ preceding a small switch $v$ exists
  since $f_0$ being non-rectangular implies that there is at least one small switch.
  
  Let $\bar{G}$ be the resulting digraph where all faces are rectangular. 
  By construction, $\bar{G}$ has $G$ as topological minor and size in $\Oh(n)$.  
  Furthermore, since the boundary of every face of $G$ was traversed only twice,
  this refinement algorithm runs in $\Oh(n)$ time.
\end{proof}

We can now prove the main theorem of this section.

\begin{theorem} \label{clm:plane}
Let $G$ be an upward plane digraph with $n$ vertices. Then the following statements are equivalent.
\begin{enumerate}[leftmargin=*,label=(F\arabic*)]
  \item \label{case:f1} $G$ admits an upward planar $2$-slope drawing.
  \item \label{case:f2} $G$ admits a $2$-slope representation $U_G$.
  \item \label{case:f3} $G$ contains no bad edge.
\end{enumerate} 
Moreover, there exists an $\Oh(n)$-time algorithm that tests if $G$ satisfies \ref{case:f3}
and constructs an upward planar 2-slope drawing of $G$ in the affirmative case.
\end{theorem}
\begin{proof}
  Note that \ref{case:f1} implies \ref{case:f2} and \ref{case:f2} implies \ref{case:f3}. 
  We first show that \ref{case:f3} implies \ref{case:f2} and then how to construct a drawing \ref{case:f1} from \ref{case:f2}. 
 
  Whether $G$ contains a bad edge can easily be checked in $\Oh(n)$ time.
  Suppose it does not and thus satisfies \ref{case:f3}.
  Construct a consistent 2-slope assignment for $G$ as follows.
  Go through all edges of $G$ (in any order). 
  For an edge $e$, if it is a left (right) incoming edge, assign to it slope~$\sR$ (resp.~$\sL$).
  Otherwise, if it is a left (right) outgoing edge, assign to it slope~$\sL$ (resp.~$\sR$).
  We claim that since $e$ is not a bad edge, there is no conflict.
  Assume otherwise, namely, that $e = (u, v)$ gets, say, slope $\sL$ from $u$ and slope~$\sR$ from~$v$. 
  However, then $u$ must be the left outgoing edge at $u$ and the left incoming edge at $v$,
  which makes $e$ a bad edge thus contradiction our assumption.  
  If~$e$ is both a sole incoming and a sole outgoing edge, assign it an arbitrary slope, say $\sL$. 
  Together with the already given upward planar embedding of $G$, 
  this slope assignment yields a 2-slope representation~$U_G$ of $G$.
  
  Next, we construct an upward planar 2-slope drawing of $G$.
  Use \cref{clm:plane:refinement} to obtain a $2$-slope representation $U_{\bar G}$ of an upward planar digraph $\bar G$ in which every face is rectangular
  in $\Oh(n)$ time.
  Note that rotating $\bar G$ clockwise by 45$^\circ$,
  makes the slope $\sL$ vertical and the slope $\sR$ horizontal and we get an orthogonal representation of a graph where every face is rectangular.
  Therefore we can apply the linear-time compaction algorithm by Di Battista \etal~\cite[Theorem~5.3]{DETT99}.
  This algorithm assigns edge lengths and computes coordinates 
  while handling the vertical and orthogonal direction of a orthogonal representation independently.
  Hence, applying the algorithm to $U_{\bar G}$, the edges with slopes $\sR$ and $\sL$ are handled independently and keep their slopes.
  (Note that a, say, vertical edge $(u, v)$ with length $c$ in the orthogonal drawing 
  corresponds to $v$ being $c$ units above and to the left of $u$ in an upward planar 2-slope drawing.) 
  As a result, we get an upward planar 2-slope drawing of $\bar G$, which we can reduce to a drawing of $G$.
  Since the three steps run in $\Oh(n)$ time each, the claim on the running time follows.
\end{proof}

Suppose we have an upward plane digraph $G$ with $k$ bad edges. 
Since $G$ admits no upward planar 2-slope drawing, it is natural to ask whether $G$ admits an upward planar 1-bend 2-slope drawing.
In particular, if such a drawing exist, is it enough to bend only the $k$ bad edges? 
Using \cref{clm:plane} we can answer this question affirmatively.

\begin{corollary} \label[corollary]{clm:plane:bends}
Let $G$ be an upward plane digraph with $n$ vertices, maximum in- and outdegree at most two, and with $k$ bad edges.
Then $G$ admits an upward planar 1-bend 2-slope drawing with $k$ bends.
Moreover, without changing the embedding, this is the minimum number of bends that can be achieved.
\end{corollary}
\begin{proof}
  Subdivide every bad edge once to obtain a graph $G'$.
  Then apply \cref{clm:plane} to obtain an upward planar 2-slope drawing of $G'$.
  Since a bad edge $e$ of $G$ is neither a sole incoming nor a sole outgoing edge, 
  the two edges obtained from $e$ in $G'$ have different slopes.
  Hence, by turning every subdivision vertex into a bend we get an upward planar 1-bend 2-slope drawing of $G$.
  
  Since even a single bad edge obstructs a 2-slope representation 
  and since bending a non-bad edge clearly does not eliminate any other bad edges either,
  it follows that $k$ bends are also necessary. 
\end{proof}

Note that $G$ may admit an upward planar 1-bend 2-slope drawing with less or no bends if the embedding is changed.

\section{Variable embedding} \label{sec:variable} 
In this section we consider the problem of whether a given upward planar digraph $G$ of a particular graph class
admits an upward planar $2$-slope drawing under any upward planar embedding.
We start with two general observations, 
before we consider the class of single-source digraphs, 
where upward planarity can be tested in linear time,
and then continue with the more complex classes of series-parallel digraphs and general digraphs. 

Let $G$ be an upward planar digraph with maximum in- and outdegree at most two. Suppose~$G$ contains a leaf $\ell$.
Note that removing $\ell$ from $G$ does not change whether $G$ admits an upward planar 2-slope drawing.
Moreover, we may reduce $G$ to a digraph without leaves, obtain a drawing of the reduced digraph (if possible), 
and then add the leaves to obtain a $2$-slope drawing of $G$. 
Removing and later restoring leaves takes only linear time.

While leaves are no obstruction, transitive edges are. 
However, note that not all bad edges have to be transitive edges; see \cref{fig:singleSource}. 
   
\begin{observation} \label{clm:planar:transitiveEdge}
A transitive edge of an upward planar digraph $G$ is a bad edge in any upward planar embedding of $G$.
\end{observation}
\begin{proof}
  Let $e = (u, v)$ be a transitive edge of $G$. 
  By definition, there is a directed path~$P$ from $u$ to $v$ different from $e$.
  Since $P$ may not cross $e$, it enters $v$ from the same side (left or right) as it leaves $u$ in any upward planar embedding of $G$
  and hence $e$ is always a bad edge.
\end{proof}

\subsection{Single-source digraphs} \label{sec:ss} 
For a single-source digraph $G$, our idea is to first compute an arbitrary upward planar embedding of $G$ 
with the linear-time algorithm of Bertolazzi \etal~\cite{BDMT98}.
We then check whether there are any bad edges and, if so, whether they can be fixed with small changes to the embedding.
To this end, we need the following lemmata.

\begin{lemma} \label[lemma]{clm:singleSource:outerFace}
An upward planar single-source digraph $G$ contains no bad edge with respect to the outer face.
\end{lemma}
\begin{proof}
  Consider an upward planar embedding of $G$ 
  and suppose $e = (u, v)$ is a bad edge with respect to the outer face $f_0$; 
  see \cref{fig:singleSource:twoBadEdges}~(a). 
  Let $e'$ be the second incoming edge of $v$.
  Then $v$ is a local sink of $f_0$ such that if $e$ and $e'$ would be drawn with slopes $\sR$ and $\sL$ accordingly,
  then $f_0$ would have a small angle at $v$.
  However, this implies that there are two local sources ($w_1$ and $w_2$ in \cref{fig:singleSource:twoBadEdges}~(a)) of $f_0$ 
  that are also sources of $G$. This is a contradiction to $G$ being a single-source digraph.
\end{proof}

\begin{lemma} \label[lemma]{clm:singleSource:twoBadEdges}
Let $G$ be an upward planar single-source digraph with maximum in- and outdegree at most two.
Then in any upward planar embedding of $G$, there are at most two bad edges with respect to the same face. 
\end{lemma}
\begin{proof}
  Let $G$ be an upward plane single-source digraph and let $f$ be a face of $G$. 
  Note that a bad edge $e$ with respect to $f$ is incident to a local sink $v$ of a face $f$ 
  where $v$ lies between its two incoming edges in a counterclockwise traverse of the boundary of $f$.
  In other words, in an upward planar drawing of $G$, $f$ would have a (small) angle of less than 180$^\circ$ at $v$.
  If there are three or more bad edges with respect to $f$, 
  then $f$ contains at least two such local sinks (like $v_1$ and $v_2$ in \cref{fig:singleSource:twoBadEdges}~(b)).
  There is then at least one local source $u$ for $f$ that would span a large angle in $f$, i.e., more than 180$^\circ$.
  However,~$u$ is also a source of $G$ and since it is not the source for the outer face, it is not the only source of $G$.
  This is a contradiction to $G$ being a single-source digraph. 
  \Cref{fig:singleSource:twoBadEdges}~(c) gives an example of a face with two bad edges.
\end{proof}

\begin{figure}[htb]
  \centering
  \includegraphics{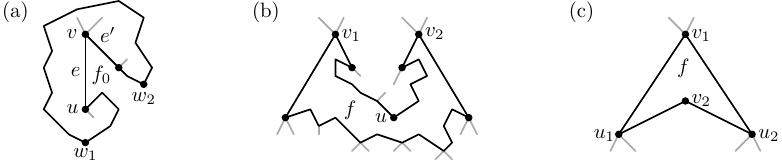}
  \caption{(a) A bad edge with respect to the outer face~$f_0$ implies the existence of at least two sources~$w_1$ and~$w_2$;
  (b)~a face~$f$ that has four bad edges of which two are adjacent to the local sink~$v_1$ of~$f$ and 
  two are adjacent to the local sink~$v_2$ of~$f$. However,~$f$ also contains a source~$u$; 
  (c)~a face~$f$ with two bad edges~$(u_1, v_1)$ and~$(u_2, v_1)$ that can be a face of a single source digraph.}
  \label{fig:singleSource:twoBadEdges} 
\end{figure}

\begin{lemma} \label[lemma]{clm:singleSource:fixBadEdge}
Let $G$ be an upward plane single-source digraph with maximum in- and outdegree at most two
and with bad edges $\set{e_1, e_2, \ldots, e_k}$.
Let~$e_1$ be a bad edge with respect to face~$f$.
Deciding whether~$G$ admits an upward planar embedding with bad edges~$\set{e_2, \ldots, e_k}$
can be done in~$\Oh(\abs{f})$ time.
\end{lemma}
\begin{proof}
  Let~$e_1 = e = (u, v)$, let~$e_u$ be the second outgoing edge of~$u$, 
  and let~$e_v$ be the second incoming edge of~$v$.
  Note that~$e_v$ and~$e_u$ are also on the boundary of~$f$
  and that by \cref{clm:singleSource:outerFace},~$f$ is not the outer face.
  We claim that~$e$ cannot be a bridge. 
  Assume otherwise and let~$G_u$ and~$G_v$ be the components of~$G \setminus \set{e}$ that contain~$u$ and~$v$, respectively.
  Note that both~$G_u$ and~$G_v$ have a source each and~$v$ cannot be the source of~$G_v$. 
  This implies that~$G$ has two sources, which is a contradiction to~$G$ being a single-source digraph. 
  Let~$f'$ be the second face with~$e$ on its boundary, which exists since~$e$ is not a bridge. 
  Without loss of generality, assume that~$f'$ is to the left and~$f$ to the right of~$e$.
  
  For $G$ to admit an upward planar embedding where $e$ is not a bad edge,
  it must be possible to change, without loss of generality, the order of $e_u$ and $e$ at $u$,
  that is, $e_u$ and $e$ need to become the left and right outgoing edges of~$u$, respectively.
  If $e_u$ is a bridge (and thus has $f$ as left and right face),
  then we can swap $e$ and $e_u$ at $u$
  and flip the component that does not contain $u$ of $G \setminus \set{e_u}$; see \cref{fig:singleSource}~(a).
  We can find out whether $e_u$ is a bridge with a single traverse of $f$.
  Furthermore, clearly, this flip does not introduce a new bad edge.
  
  \begin{figure}[htb]
  \centering
  \includegraphics[page=1]{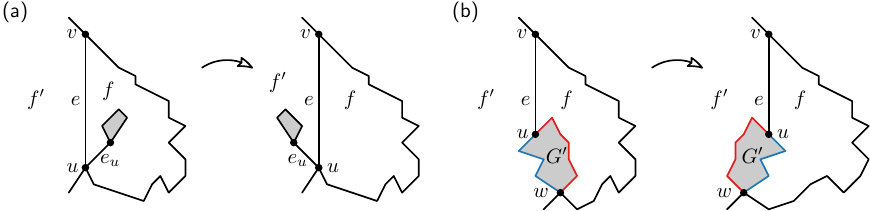}
  \caption{Two scenarios where a bad edge~$e$ with respect to~$f$ in a single-source digraph can be repaired: 
	  (a)~$e_u$ is a bridge and can be flipped from $f$ to $f'$;
	  (b)~the digraph~$G'$ between $f$ and $f'$ and between the split pair~$\set{w, u}$ can be flipped.} 
  \label{fig:singleSource}
  \end{figure}
  
  Otherwise, if~$e_u$ is not a bridge, 
  observe that we need to flip a subgraph~$G'$ of~$G$ 
  that contains~$e_u$ and that is enclosed by~$f'$ and~$f$; see \cref{fig:singleSource}~(b).
  For such~$G'$ to exist,
  there must be a vertex~$w$ that forms a split pair with~$u$,
  and that has~$f$ to the right and~$f'$ to the left.
  It can be seen as a split pair~$\set{u, w'}$
  without this property does not yield a flippable digraph~$G'$ (as in \cref{fig:singleSourceBad}~(a))
  and if no such a split pair exists at all, then the triconnectedness implies 
  that~$e_v$,~$e$, and~$e_u$ always lie on the boundary of the same face (as in \cref{fig:singleSourceBad}~(b)). 
  Assuming now that such~$w$ exists, 
  we can define~$G'$ as the digraph consisting of all split components of~$G$ 
  with respect to~$\set{w, u}$ that do not contain~$v$.
  To repair~$e$, we flip~$G'$ as shown in \cref{fig:singleSource}~(b), that is, 
  we reverse the order of incoming and outgoing edges for each vertex in~$G'$
  except for~$w$, where we only reverse the order of the outgoing edges.
  Note that the flip cannot introduce a new bad edge since~$w$ is not a local sink or local source of~$f$ or~$f'$.
  Furthermore, note that such~$w$ precedes~$u$, since 
  otherwise both~$G'$ and~$G \setminus G'$ would contain a source (that is not~$w$) each,
  which contradicts~$G$ being a single-source digraph; see \cref{fig:singleSourceBad}~(c).
  Hence, we may find~$w$ with a traverse of~$f$.
  
  \begin{figure}[htb]
  \centering
  \includegraphics[page=2]{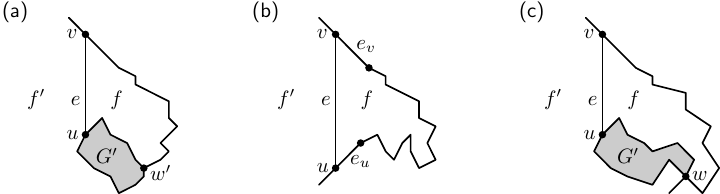}
  \caption{Scenarios where a bad edge~$e$ with respect to~$f$ in a single-source digraph
  cannot be repaired (a)~since~$G'$ cannot be flipped
  or (b)~since there is no split pair~$\set{w, u}$;  
  (c)~here $G'$ can be flipped, but the digraph is not a a single-source digraph.} 
  \label{fig:singleSourceBad}
  \end{figure}
  
  Lastly, we show that changing the order of $e$ and $e_v$ at $v$
  is possible only when the case above applies where the order of $e$ and $e_u$ at $u$ can be changed.
  To begin with, note that $e_v$ cannot be a bridge, since~$G$ is a single-source digraph.
  Hence, we would need to flip again a digraph $G'$ that contains~$e_v$ and that is enclosed by $f$ and $f'$.
  Such $G'$ would imply a split pair $\set{w', v}$,
  where, however, $w'$ would also serve as split pair $\set{w', u}$ as in \cref{fig:singleSource}~(b).
  
  Since we can check whether $e_u$ is a bridge or
  find a suitable $w$ with a single traverse of $f$,
  the claim on the running time holds.
\end{proof}

From \cref{clm:singleSource:fixBadEdge},
we get that an upward plane single-source digraph admits an upward planar 2-slope drawing
(possibly for a different upward planar embedding)
if every bad edge can be fixed. We may thus check every bad edge and perform the necessary flips.
Since digraphs to flip may be nested, executing simply one flip after the other could be costly.
We now show how to keep the running time linear. 

\begin{theorem} \label{clm:singleSource:algorithm}
Let~$G$ be an upward planar single-source digraph with $n$ vertices.
An upward planar 2-slope drawing of $G$ can be computed in $\Oh(n)$ time, if one exists.
\end{theorem}
\begin{proof}
As noted above, we may assume that~$G$ does not contain a leaf.
A linear-time algorithm to compute an upward planar 2-slope drawing of~$G$ then works as follows.
First, compute an upward planar embedding of~$G$ with the algorithm of Bertolazzi \etal~\cite{BDMT98} in~$\Oh(n)$ time.
Second, to identify all bad edges in~$\Oh(n)$ time it suffices to traverse the boundary of each face
since every bad edge is bad with respect to exactly one face. 
Third, for each bad edge~$e = (u, v)$ with respect to a face~$f$, 
check whether it can be repaired with \cref{clm:singleSource:fixBadEdge}.
To keep track of where we have to flip the edge order at vertices, 
we use two types of markers.
A green marker at a vertex~$x$ indicates
that the outgoing edges of~$x$ should be swapped
and that both incoming and outgoing edges of the vertices ``above''~$x$
should be reversed.
We thus mark~$u$ green if~$e$ can be repaired because the respective edge~$e_u$ is a bridge
and we mark~$w$ green if~$e$ can be repaired because of a respective split pair~$\set{w, u}$.
Furthermore, we mark~$e$ red to tells us to stop reversing edge orders.
Both the check and the marking can be done in~$\Oh(\abs{f})$ time.
Since by \cref{clm:singleSource:twoBadEdges} any face contains at most two bad edges, 
this step takes overall also only~$\Oh(n)$ time.

Suppose every bad edge can be fixed.
Then run a BFS on~$G$ that starts at the source.
During the traversal, remember along each path the number of green marked vertices
minus the number of encountered red edges. 
Then for each vertex~$v$, use the parity of this number to decide whether its edge orders have to be reversed;
if odd, then reverse, and if even, then do not.
Vertices marked green need to be handled appropriately.
This takes again only linear time and as a result we get an upward planar embedding of~$G$ that admits a 2-slope representation~$U_G$.
Finally, apply the algorithm from \cref{clm:plane} on~$U_G$ to compute an upward planar 2-slope drawing of~$G$.
\end{proof}

Note that in \cref{clm:singleSource:fixBadEdge}, whether a bad edge is bad in any upward planar embedding of $G$ is independent from
whether another edge is bad in any upward planar embedding of $G$.
Hence, we can also use \cref{clm:singleSource:algorithm} and \cref{clm:plane:bends} to minimise the number of bends
in a 1-bend 2-slope drawing of $G$.

\begin{corollary} \label[corollary]{clm:planar:singleSource:oneBend}
Let $G$ be an upward planar single-source digraph with $n$ vertices and maximum in- and outdegree two.
An upward planar 1-bend 2-slope drawing of $G$ with the minimum number of bends can be computed in $\Oh(n)$ time.
\end{corollary}

\subsection{SPQR-trees and upward spirality} \label{sec:spqr} 
Didimo \etal~\cite{DGL10} described algorithms to compute upward planar embeddings of biconnected series-parallel and general digraphs.
They then use a result from Healy and Lynch~\cite{HL07} about combining biconnected blocks of upward planar digraphs 
to get rid of the biconnectivity condition.
We follow their approach closely.
In particular, we also use the notions of SPQR-trees and upward spirality (with the latter tailored to our needs)
on which Didimo et al.'s approach heavily relies on and tailor them to our needs.
We refer to Didimo \etal~\cite{DGL10} for the precise definition of SPQR-trees (for undirected graphs and then derived for digraphs),
though recall the main concepts in this section. 

Let $G$ be a biconnected digraph and let $e = (s, t)$ be any edge of $G$ called \emph{reference edge}.
An \emph{SPRQ-tree} $T$ of $G$ with respect to $e$ represents a decomposition of $G$ with respect to its triconnected components~\cite{DT89,GM01}.
As such, it also represents all planar embeddings of $G$ with $e$ on the outer face.
Starting with the split pair $\set{s, t}$, the decomposition is constructed recursively on the split pairs of $G$.
More precisely, $T$ is a rooted tree where every node is of type S, P, Q, or R:
Q-nodes represent single edges, S-nodes and P-nodes represent \emph{series components} and \emph{parallel components}, and
R-nodes represent triconnected (\emph{rigid}) components; see \cref{fig:SPQRexample} for an example. 
Each node~$\mu$ in~$T$ has associated a biconnected multigraph~$\skel(\mu)$, called the \emph{skeleton} of~$\mu$, 
in which the children of~$\mu$ and its reference edge are represented by a virtual edge. 
The root of~$T$ is a Q-node representing~$(s, t)$. The child of~$\mu$ is now defined recursively as follows:
\begin{description}
  \item[Trivial case.] If $G$ consists of exactly two parallel edges between $s$ and $t$, 
  then $\mu$ is a Q-node representing the edge $(s, t)$ parallel to the reference edge and the skeleton $\skel(\mu)$ is $G$.
  \item[Parallel case.] If the split pair $\set{s, t}$ has three or four split components $G_1, G_2, \ldots, G_{k}$ 
  (more are not possible with our degree restrictions), then $\mu$ is a P-node. 
  The skeleton $\skel(\mu)$ consists of $k$ parallel edges between $s$ and $t$ that represent
  the reference edge $G_1$ and the components~$G_i$, $2 \leq i \leq k$.
  \item[Series case.] If $\set{s, t}$ induces two split components $e = (s, t)$ and $\bar G$, then $\mu$ is an S-node.
  If $\bar G$ is a chain of biconnected components $G_1, \ldots, G_k$ with cut vertices $c_1, \ldots, c_{k-1}$ ($k \geq 2$),
  then $\skel(\mu)$ is the cycle $t$, $e$, $s$, $e_1$, $c_1$, $\ldots$, $c_{k-1}$, $e_k$, $t$ where~$e_i$ represents~$G_i$.
  \item[Rigid case.] Otherwise $\mu$ is an R-node.
  Let $\set{s_1, t_1}$, \ldots, $\set{s_k, t_k}$ be the maximal split pairs of~$G$ with respect to $\set{s, t}$.
  Let $G_i$ be the union of split components of $\set{s_i, t_i}$ except the one containing $e$.  
  Then $\skel(\mu)$ is obtained from $G$ by replacing each subgraph $G_i$ with $e_i$.
\end{description}
The skeleton of the root consists of two parallel edges, $e$ and a virtual edge representing the rest of the graph.
Each node $\mu$ of $T$ that is not a Q-node has children $\mu_1, \ldots, \mu_k$, in this order such that
$\mu_i$ is a child of $\mu$ based on $G_i \cup e_i$ with respect to $e_i$. 
The \emph{pertinent graph} $G_{\mu}$ for a node $\mu$ of $T$ represents the full subgraph of $G$ in the SPQR-tree rooted at $\mu$.
The end vertices of $e_i = (s_i, t_i)$ are called the \emph{poles} of $\mu$, of $\skel(\mu)$, and of $G_{\mu}$.
Refer to \cref{fig:SPQRexample} for an example and note that every edge is represented by a Q-node.
We assume that $T$ is in its \emph{canonical form}, that is, every S-node of $T$ has exactly two children. 
The canonical form can be derived from a non-canonical form in linear time~\cite{DGL10}.

\begin{figure}[htb]
  \centering 
  \includegraphics{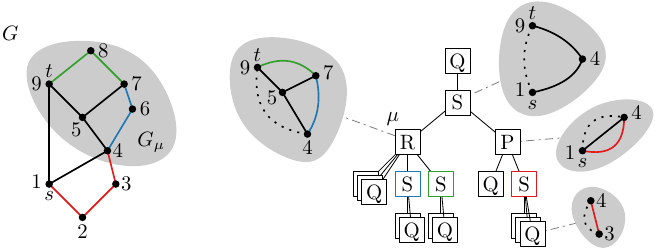}
  \caption{A digraph $G$ and its (non-canonical) SPQR-tree with respect to the reference edge $\set{s,t}$.
  The skeletons of some nodes are depicted with the virtual edge drawn dashed. 
  The pertinent graph $G_\mu$ of the R-node $\mu$ is highlighted in $G$. }
  \label{fig:SPQRexample}
\end{figure}

Assume that $G$ is equipped with an st-numbering (based on its underlying undirected graph and with respect to the reference edge $\set{s,t}$).
Let $\mu$ be a node of $T$ with poles $u$ and $v$ such that $u$ precedes $v$ in the st-numbering.
Then $u$ is the \emph{first pole} and $v$ is the \emph{second pole} of $\mu$.

Assume that a pertinent graph $G_\mu$ of a node $\mu$ of $T$
admits an upward planar 2-slope drawing and let $U_{G_\mu}$ be a $2$-slope representation of $G_\mu$.
Let~$u$ and $v$ be the poles of $\mu$ and let $w \in \set{u, v}$.
The \emph{pole category} $t_w$ of the pole $w$ is the way the edges that lie in the outer face of ${G_\mu}$ 
are incident to $w$ and which slopes they got assigned under $U_{G_\mu}$. 
Note that edges incident to $w$ that lie in the interior of ${G_\mu}$ do not affect the pole category of~$w$.
The sixteen possible pole categories of split components are shown in \cref{fig:poleCategories}.
Two poles~$w$ and $w'$ with pole category $t_w$ and $t_w'$, respectively, are \emph{compatible}
if~$t_w$ and~$t_w'$ can be combined into a pole category of higher degree. 
For example, combining iR and iL gives~iRiL.

\begin{figure}[htb]
  \centering 
  \includegraphics{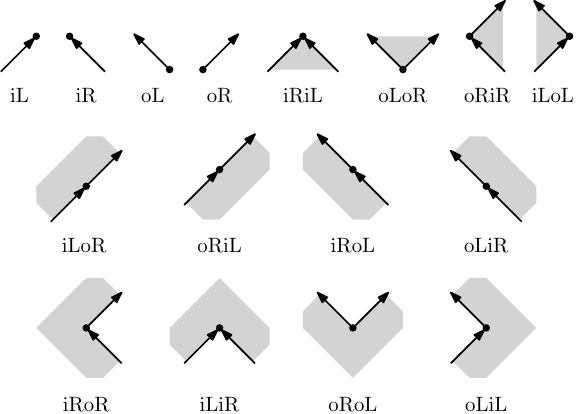}
  \caption{Pole categories in a 2-slope representation; 
  gray areas indicate where the interior of the graph lies.}
  \label{fig:poleCategories}
\end{figure}

Next we define upward spirality, which measures how much a split component is ``rolled up''.
The pertinent graph~$G_\mu$ of a node~$\mu$ of~$T$ may have 2-slope representations with different spirality.
Let~$U_{G_\mu}$ be a 2-slope representation of~$G_\mu$ with first pole~$u$ and second pole~$v$.
Let~$P$ be an undirected path in~$U_{G_\mu}$. 
Two subsequent edges~$\set{x, y}$ and~$\set{y, z}$ of~$P$ define a \emph{right (left) turn} 
if~$P$ makes a 90$^{\circ}$ clockwise (resp.\ counterclockwise) turn at~$y$ according to~$U_{G_\mu}$. 
Define the \emph{turn number}~$\n(P)$ of~$P$ as the number of right turns minus the number of left turns of~$P$.
Let~$P_l$ and~$P_r$ be the clockwise and counterclockwise paths from~$u$ to~$v$ along the outer face, respectively.
We define the \emph{upward~$2$-slope spirality}~$\sigma$ of~$U_{G_\mu}$ as 
$\sigma(U_{G_\mu}) = \frac{\n(P_l) + \n(P_r)}{2}\text{.}$ 
See \cref{fig:spirality} for examples. 
Note that the turn number of a path is bounded by its length.
Therefore, the number of possible values for the upward~$2$-slope spirality of~$U_{G_\mu}$ is in~$\Oh(n)$ 
(and in~$\Oh(d)$ where~$d$ is the diameter of~$U_{G_\mu}$).
For technical reasons that become apparent later, 
when we store the upward 2-slope spirality of a 2-slope representation,
we also store the values~$\n(P_l)$ and~$\n(P_r)$. 

\begin{figure}[htb]
  \centering
  \includegraphics[page=1]{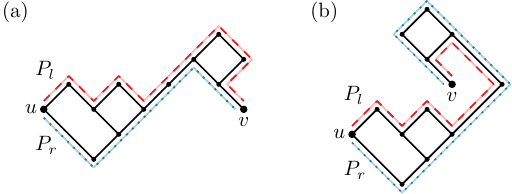}
  \caption{Two different $2$-slope representations of the same digraph for different upward planar embeddings: 
  (a) The paths $P_l$ and $P_r$ have turn number $1$ and $0$, and so the upward $2$-slope spirality is $0.5$;
  (b) the paths $P_l$ and $P_r$ have turn number $-3$ and $-4$, and so the upward $2$-slope spirality is $-3.5$}
  \label{fig:spirality}
\end{figure}

Let $\mu$ be a node of $T$ with first pole $u$ and second pole $v$.
A \emph{feasible tuple} of $\mu$ is a tuple $\tau_\mu = \langle U_{G_{\mu}}, \sigma, t_u, t_v  \rangle$
where $U_{G_{\mu}}$ is an upward 2-slope representation of $G_\mu$ with pole categories $t_u$ and $t_v$ 
and upward 2-slope spirality $\sigma$. 
Two feasible tuples of $\mu$ are \emph{spirality equivalent}
if they have the same upward $2$-slope spirality and pole categories. 
A \emph{feasible set} $\cF_\mu$ of $\mu$ is a set of all feasible tuples of $\mu$
such that there is exactly one representative tuple for each class of spirality equivalent tuples of $\mu$.

\subsection{Series-parallel digraphs} \label{sec:sp} 
Let $G$ be a biconnected series-parallel digraph, that is, an SPQR-tree (or rather SPQ-tree) of $G$ contains no R nodes.
Our goal is to test the existence of an upward planar embedding that does not contain a bad edge and thus the existence of a 2-slope representation of $G$.
The idea of the algorithm is as follows.
Pick a reference edge $e$ and construct the respective SPQR-tree $T$ of $G$.
In a post-order traversal of $T$, compute the feasible set of each node.
This is straightforward for Q-nodes. 
For an S- or P-node $\mu$, try to combine feasible tuples of the children of $\mu$ to feasible tuples of $\mu$.
The pole categories and upward $2$-slope spirality values make it easier to check
whether a composition admits an upward planar $2$-slope representation.
If this leads to a non-empty feasible set for the root of $T$, we can construct a drawing. 
Otherwise, we try again with another reference edge.
 
Let $\mu$ be a Q-node of $T$ with first pole $u$ and second pole $v$
and suppose $G_\mu$ is the edge $(u, v)$; the case where $G_\mu$ is $(v, u)$ is handled analogously.
Then there are only two upward $2$-slope representation of~$G_\mu$, 
namely when $(u, v)$ is assigned the slope $\sL$ or the slope $\sR$. Therefore, $\cF_\mu$ has size two. 
The following two lemmata show how to compute the feasible set of an S-node and a P-node of $T$.

\begin{lemma} \label[lemma]{clm:Snode}
Let $G$ be a biconnected digraph with $n$ vertices and $T$ be an SPQR-tree of $G$.
Let $\mu$ be an S-node of $T$ with children $\mu_1$ and $\mu_2$. 
Given the feasible sets $\cF_{\mu_1}$ and $\cF_{\mu_2}$, 
the feasible set $\cF_{\mu}$ can be computed in $\Oh(n^2)$ time.
\end{lemma}
\begin{proof}
	To compute a feasible tuple $\tau$ of $\cF_{\mu}$ we check for all pairs
	$\tau_1 \in \cF_{\mu_1}$ and $\tau_2 \in \cF_{\mu_2}$ whether they can be combined. 
	More precisely, let $u$ be the first and $v$ the second pole of $\mu$; refer to \cref{fig:seriesComposition}.
	Let $u_i$ be the first and $v_i$ be the second pole of $\mu_i$, $i \in \set{1, 2}$.
	Without loss of generality, assume that $u = u_1$, $v_1 = u_2$, and $v_2 = v$.
	For each pair of feasible tuples $\tau_1 = \langle U_{G_{\mu_1}}, \sigma_1, t_{u_1}, t_{v_1}  \rangle \in \cF_{\mu_1}$
	and $\tau_2 = \langle U_{G_{\mu_2}}, \sigma_2, t_{u_2}, t_{v_2}  \rangle \in \cF_{\mu_2}$
	check whether $t_{v_1}$ and $t_{u_2}$ are compatible. 
	In other words, check whether $G_{\mu_1}$ and $G_{\mu_2}$ can be plugged together at $v_1 = u_2$ 
	under the slope assignments of $U_{G_{\mu_1}}$ and $U_{G_{\mu_2}}$ and with the reference edge on the outer face.
	
	If the pole categories are compatible, 
	the feasible tuple~$\tau = \langle U_{G_{\mu}}, \sigma, t_{u}, t_{v}  \rangle$ is given by~$t_u = t_{u_1}$, $t_v = t_{v_2}$, 
	$U_{G_\mu}$ as the series composition of~$U_{G_{\mu_1}}$ and~$U_{G_{\mu_2}}$ at the common vertex~$u_2 = v_1$, 
	and where~$\sigma$ can be computed as follows. 
	For~$i \in \set{1, 2}$, let~$e_l^i$ and~$e_r^i$ be the edges of~$U_{G_{\mu_i}}$ that are incident to~$u_2 = v_1$ and lie on the
	clockwise and counterclockwise path from~$u$ to~$v$ along the outer face, respectively; see \cref{fig:seriesComposition}~(b).
	Note that~$e_l^i$ may coincide with~$e_r^i$.
	For $j \in \set{l, r}$, let~$\alpha_j$ be $-1$, $1$, or $0$ depending on whether~$e_j^1$ and~$e_j^2$ make a left, right, or no turn, respectively. 
	The upward~$2$-slope spirality of~$U_{G_\mu}$ is $\sigma = \sigma_1 + \sigma_2 + \frac{\alpha_l + \alpha_r}{2}$
	(compare to Lemma 6.4 by Didimo \etal~\cite{DGL10}).
	Store $\tau$ in~$\cF_\mu$ if~$\cF_\mu$ contains no feasible tuple with spirality equivalent to $\tau$.
	Since $\cF_{\mu_1}$ and $\cF_{\mu_2}$ have~$\Oh(n)$ tuples, $\cF_{\mu}$ can be computed in $\Oh(n^2)$ time.  
\end{proof}

\begin{figure}[t]
  \centering
  \includegraphics{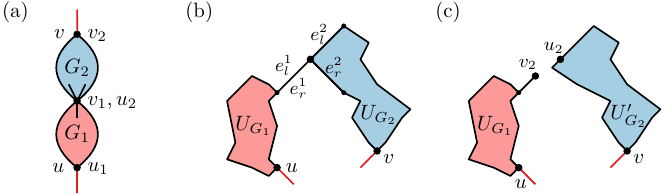}
  \caption{Illustration of a series composition of split components $G_1$ and $G_2$:
  (b) $U_{G_1}$ and $U_{G_2}$ are compatible and the upward $2$-slope spirality can be computed based on $\sigma(U_{G_1})$, $\sigma(U_{G_2})$, 
  and the turns at $e_l^1$ and $e_l^2$ and at $e_r^1$ and $e_r^2$;
  (c) $U_{G_1}$ and $U_{G_2}'$ are not compatible.}
  \label{fig:seriesComposition}
\end{figure}

\begin{lemma} \label[lemma]{clm:Pnode}
Let $G$ be a biconnected digraph with $n$ vertices and $T$ be an SPQR-tree of $G$.
Let $\mu$ be a \mbox{P-node} of $T$ with children $\mu_1, \ldots, \mu_k$ with $k \leq 4$. 
Given the feasible sets $\cF_{\mu_1}, \ldots, \cF_{\mu_k}$, 
the feasible set $\cF_{\mu}$ can be computed in $\Oh(n)$ time.
\end{lemma}
\begin{proof}
	Let $u$ be the first and $v$ the second pole of $\mu$ (and its children).
	Since $u$ and $v$ have at most degree four in $G$, 
	$\mu$ has at most four children (but at least two).
	We first consider the case where $\mu$ has three children.
	The cases where $\mu$ has two or four children are discussed at the end of this proof. 

	For $i \in \set{1, 2, 3}$, let~$G_i$ be the pertinent digraph of~$\mu_i$.
	Note that~$u$ and~$v$ have degree one in~$G_i$.
	We can thus define $e_i$ and $e_i'$ as the edges of $G_i$ incident to $u$ and $v$, respectively.
	(If $\mu_i$ is a Q-node, then $e_i = e_i'$.)
	Let $e_u$ and $e_v$ be the edges of $G$ incident to $u$ and $v$, respectively, that lie outside of $G_\mu$.
	(Note that $e_u = e_v$ is only possible in the final composition, where $e_u$ is then also the reference edge.)
	We want to construct a $2$-slope representation of~$G_\mu$ 
	such that the order of $e_1, e_2, e_3, e_u$ at $u$ 
	is the reverse order of $e_1', e_2', e_3', e_v$ at $v$; see \cref{fig:parallelComposition}~(a).
	Furthermore, the half edges representing $e_u$ and~$e_v$ have to be on the outer face.
	
	Suppose we pick a feasible tuple $\tau_1 = \langle U_{G_{1}}, \sigma_1, t_{u_1}, t_{v_1}  \rangle \in \cF_{\mu_1}$.
	Then $t_{u_1}$ and $t_{u_2}$ restrict the choices of pole categories of compatible upward $2$-slope representations of~$G_2$ and $G_3$.
	Likewise, $\sigma_1$ restricts these choices further; see \cref{fig:parallelComposition}~(b) and (c).
	More precisely, we observe that the upward $2$-slope spirality $\sigma_2$ and $\sigma_3$ of compatible $U_{G_2}$ and~$U_{G_3}$ 
	differ from~$\sigma_1$ by at least two and at most six (compare to Lemma~6.6 by Didimo \etal~\cite{DGL10}).
	Therefore, when trying all possible permutations of $G_1$, $G_2$, and $G_3$, 
	we only have to consider $\Oh(1)$ feasible tuples in $\cF_{\mu_2}$ and $\cF_{\mu_3}$ instead of iterating over complete sets.
	Storing the feasible tuples in a hash table based on the spirality and pole categories 
	we can find compatible $\tau_2$ and $\tau_3$ in constant time.
	Similar to the series-composition, we can compute the upward $2$-slope spirality of the resulting upward $2$-slope representation $U_{G_\mu}$
	based on its categories and~$\sigma_i$, $i \in \set{1, 2, 3}$.
	More precisely, for this purpose we not only stored each $\sigma_i$ but also the respective values $\n(P_l)$ and $\n(P_r)$ to compute them.
	Therefore, we can take the appropriate values from the two 2-slope representations of $U_{G_{i}}$, $i \in \set{1, 2, 3}$,
	that are on the outer face. 
	Overall, we get that by iterating once over $\cF_{\mu_1}$ we can construct all feasible tuples of $\cF_\mu$ in $\Oh(n)$ time. 
	
	The case where $\mu$ has four children is only possible in the final composition with the reference edge $e$.
	Here the algorithms works along the same line with the difference that the half edges $e_u$ and $e_v$ are replaced with $e$.
	The case where $\mu$ has two children, works analogously with the difference that for some feasible tuples of $U_{G_1}$
	there can now be more compatible upward spirality values for $U_{G_2}$ than above. 
	However, these are still $\Oh(1)$ many and the running time is thus not affected.
\end{proof}

\begin{figure}[htb]
  \centering
  \includegraphics{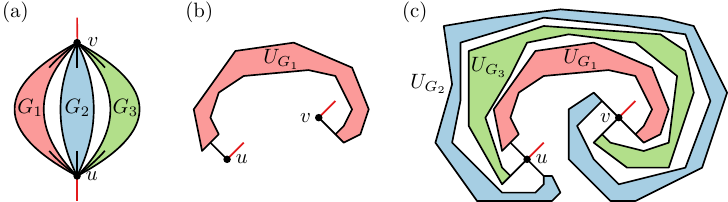}
  \caption{Illustration of a parallel compositions with three children and how picking an upward $2$-slope representation of $G_1$ 
  enforces the spirality of representation of $G_2$ and $G_3$.}
  \label{fig:parallelComposition}
\end{figure}

Lastly, for the root composition we check in $\Oh(n)$ time whether the root of $T$, which is a Q-node representing $e$, 
can be combined with a feasible tuple of its child. 
In the affirmative case, we obtain an upward $2$-slope representation of $G$.
With all compositions described, we can now prove the main theorem of this section. 

\begin{theorem} \label{clm:seriesParallel}
Let $G$ be a biconnected series-parallel digraph with $n$ vertices.
There exists an $\Oh(n^4)$-time algorithm that tests if $G$ admits an upward planar 2-slope drawing and, if so,
that constructs such a drawing.
\end{theorem}
\begin{proof}
Let $e$ be an edge of $G$. Compute the (canonical) SPQR-tree $T$ with respect to~$e$ of~$G$, which can be done in $\Oh(n)$ time~\cite{GM01,DGL10}.
In a post-order traversal of $T$, the algorithm computes the feasible set for every node of $T$. 
If the algorithm arrives at a node with an empty feasible set, its starts with another reference of $G$.
Otherwise, the algorithm stops when it has constructed a feasible tuple for $G$.
We can then use \cref{clm:plane} to construct an upward planar 2-slope drawing.
For one reference edge, this takes at most $\Oh(n^2)$ time per node by \cref{clm:Snode,clm:Pnode}.
and since the size of $T$ is linear in $n$, at most $\Oh(n^3)$ time in total.
The total running time is thus in $\Oh(n^4)$.
\end{proof}

In \cref{sec:nonbi} we explain how to handle non-biconnected series-parallel digraphs.

\subsection{Biconnected digraphs} \label{sec:digraph} 
We extend the algorithm for biconnected series-parallel digraphs to general biconnected digraphs following again Didimo \etal~\cite{DGL10}. 
The upward planarity check is again combined with finding a 2-slope representation. 
Let $G$ be a biconnected digraph.
Let~$T$ be the SPQR-tree of $G$ with respect to a reference edge $e$.
The algorithm computes again the feasible sets of the nodes of $T$ in a post-order traversal.
For Q-, S-, or P-nodes this works as before. 
Recall that to compute a feasible tuple of an S-node or P-node 
it suffices to look at the pole categories and upward spirality of its children.
For R-node this connection is not as clear and we rely thus on a brute-force approach.
More precisely, we compute the feasible set by considering all possible combinations of tuples for each virtual edge of $\skel(\mu)$ to construct $U_{G_\mu}$.
If substitutions are successful, we have to check upward planarity and the existence of bad edges.

\begin{lemma} \label[lemma]{clm:Rnode}
Let $G$ be a biconnected digraph with $n$ vertices and $T$ be an SPQR-tree of $G$.
Let $\mu$ be an R-node of $T$ with children $\mu_1, \ldots, \mu_k$. 
Let $d$ be the diameter of $G_\mu$.
Given the feasible sets $\cF_{\mu_1}, \ldots, \cF_{\mu_k}$, 
the feasible set $\cF_{\mu}$ can be computed in $\Oh(d^k n^2)$ time.
\end{lemma}
\begin{proof}
    Note that since $\skel(\mu)$ is triconnected, it has a unique planar embedding (up to mirroring),
    which we can compute in $\Oh(n)$ time. 
	Note that, by \cref{clm:plane}, if $\skel(\mu)$ contains a bad edge with respect to three non-virtual edges, 
	then $\cF_{\mu}$ is empty. Moreover, then $G$ admits no upward planar 2-slope drawing and the main algorithm can stop.
	So suppose no such bad edge exists.
	
	Construct a 2-slope representation of $G_\mu$ by substituting virtual edges with the respective 2-slope representations.
	More precisely, for all $i \in \set{1, \ldots, k}$, 
	substitute $e_i$ with the $2$-slope representation $U_{G_{\mu_i}}$ of a feasible tuple in $\cF_{\mu_i}$.
	Let $U_{G_{\mu}}'$ denote the partial upward $2$-slope representation of $G_\mu$ during this process,
	that is, $U_{G_{\mu}}'$ consists of an embedding of $G_{\mu}$ 
	and the 2-slope assignment for all edges of the $U_{G_{\mu_i}}$, for $i \in \set{1, \ldots, k}$.
	At each pole of a child $\mu_i$ in~$U_{G_{\mu}}'$ check 
	whether the $2$-slope representations of the substituted parts are conflicting.
	If this check fails, backtrack and try another feasible tuple.	
	Suppose that it is successful for all poles of all $\mu_i$.
	Then test the upward planarity of $U_{G_\mu}'$ with the flow-based upward planarity algorithm 
	for triconnected digraphs by Bertolazzi \etal~\cite{BDLM94}
	where the assignment of switches to faces is given for the substituted parts
	(derived from $U_{G_{\mu}}'$).
	This flow-based algorithm runs in $\Oh(n^2)$ time.
	
	If $U_{G_\mu}'$ is upward planar, check whether $U_{G_\mu}'$ 
	contains any bad edge and, if not, extend $U_{G_\mu}'$ to a 
	$2$-slope representation $U_{G_{\mu}}$ of $G_\mu$ (compare to Lemma~8.2 by Didimo \etal~\cite{DGL10}).
	Suppose there is an edge $(x, y)$ on the outer face of $U_{G_{\mu}}$ 
	that is the sole outgoing edge of $x$ and the sole incoming edge of $y$.
	Note that the choice of slope for~$(x, y)$ does not influence the spirality of $U_{G_{\mu}}$,
	since the angles formed at $x$ and $y$ always add up to the same value; see \cref{fig:spiralityChoice}.
	Lastly, compute the upward $2$-slope spirality of $U_{G_{\mu}}$ in $\Oh(n)$ time
	to obtain a feasible tuple for $\cF_\mu$.
	By backtracking and trying the remaining feasible tuples of the $\mu_i$,
	we complete the computation of $\cF_\mu$.

	\begin{figure}[htb]
	  \centering
	  \includegraphics[page=2]{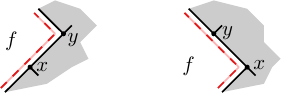}
	  \caption{In an R-node $\mu$, the slope choice for an edge $(x, y)$ that is the sole outgoing edge at $x$ and sole incoming edge at $y$
	  along the outer face does not effect the upward spirality $U_{G_{\mu}}$.}
	  \label{fig:spiralityChoice}
	\end{figure}

	Note that $d$ is an upper bound on the upward $2$-slope spirality of a split component
	and thus each $\mu_i$ has $\Oh(d)$ feasible tuples. 
	Therefore the algorithm tries at most $\Oh(d^k)$ combinations of feasible tuples of the $k$ children of $\mu$.
	Hence the feasible set of an R-node~$\mu$ can be computed in $\Oh(d^k n^2)$ time, 
	where the factor $\Oh(n^2)$ comes from the flow based upward planarity test.
\end{proof}

Let $d$ denote the maximum diameter of a split component of $G$
and let $t$ denote the number of nontrivial triconnected components of $G$, 
i.e., series components, parallel components, and rigid components.
Didimo \etal~\cite{DGL10} have shown that in \cref{clm:Rnode}
we can also bound the running time with $\Oh(d^t n^2)$
and further, in total, compute the feasible sets of all R-nodes in $\Oh(d^t t n^2)$ time.
Recall that the time needed to compute the feasible set of an S-node is bounded by the square of possible spirality values and is thus in $\Oh(d^2)$.
Since we use the canonical form of SPQR-trees there are $\Oh(n)$ S-nodes. 
Therefore, we can compute the feasible sets of all S-nodes of $T$ in $\Oh(d^2n)$ time.
Similarly, the feasible sets of all P-nodes of $T$ can be computed in $\Oh(dt)$ time. 
Lastly, iterating over all SPQR-trees of $G$ for the different choices of reference edges adds another factor of $\Oh(n)$.
If a 2-slope representation of $G$ has been found, we apply again \cref{clm:plane} to compute a drawing.
Hence, we get the following theorem.

\begin{theorem} \label{clm:digraph}
Let $G$ be a biconnected digraph with $n$ vertices.
Suppose that $G$ has at most $t$ nontrivial triconnected components, and that each split component has diameter at most $d$.
Then there exists an $\Oh(d^t t n^3 + dtn +d^2n^2)$-time algorithm that tests if $G$ admits an upward planar 2-slope drawing and, if so,
that constructs such a drawing of $G$.
\end{theorem}

\subsection{General digraphs}\label{sec:nonbi} 
So far we have seen how to test whether a biconnected digraph admits a 2-slope representation.
For a general digraph~$G$, even if each of its biconnected components, called a \emph{block}, has a 2-slope representation 
we may not be able to join the blocks; see \cref{fig:general:nonMerging}.
In fact, even if all blocks of~$G$ being upward planar does not imply that~$G$ is upward planar~\cite{HL07}. 
Nonetheless, following Healy and Lynch~\cite{HL07}, our strategy is to test for each block 
if its admits a 2-slope representation under some special conditions 
and, in the affirmative, join these representations.

\begin{figure}[htb]
  \centering
  \includegraphics{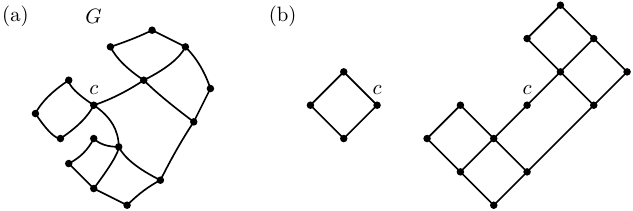}
  \caption{$G$ does not admit a 2-slope representations (a), even though its blocks do (b).}
  \label{fig:general:nonMerging}
\end{figure}

When we want to merge the representations of two blocks, we have to take two things into consideration.
Namely, we have to test whether their joined cut vertex~$c$ is on the outer face for one of the two blocks 
and whether their 2-slope representations fit together at~$c$ (just like poles and their pole categories). 
Before we get into detail on this, we recall what a block-cut tree is 
and explain how it gives a suitable order to process the blocks.    

\paragraph{Block-cut tree.}
The \emph{block-cut tree} $\cT$ of $G$ contains a vertex for each block and for each cut vertex, 
and an edge between a cut vertex $c$ and each block that contains $c$.
For $G$ to admit a 2-slope representation, 
we must be able to root $\cT$ at a block (with all edges oriented towards this root block)
such that an edge $\set{B, c}$ between a block $B$ and cut vertex $c$ is oriented towards $c$
only if $B$ admits a 2-slope representation where $c$ is on the outer face 
(see \cref{fig:blockCutTree} and compare to Lemma 3 by Chan~\cite{Cha04}). 
Note that if we can root $\cT$ at a block $B'$, then any other block $B$ has exactly one outgoing edge
and $B'$ has only incoming edges.
For a block $B$ with outgoing edge to a cut vertex $c$,
we say $B$ is a \emph{block with respect to} $c$.  
Note that multiple blocks can be a block with respect to the same cut vertex. 

\begin{figure}[htb]
  \centering
  \includegraphics{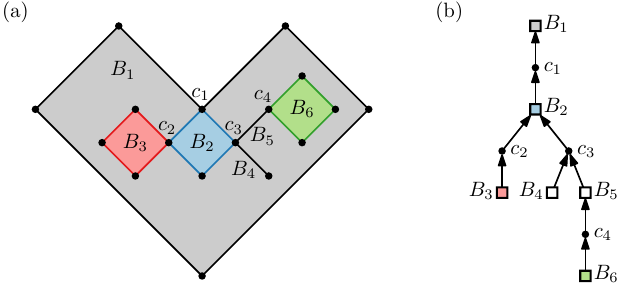}
  \caption{A digraph with six blocks and 2-slope representation (a) and the corresponding rooted block-cut tree (b).}
  \label{fig:blockCutTree}
\end{figure}

Suppose we would know how to root $\cT$. 
With a post-order traversal of $\cT$, we could then try to find 
a 2-slope representation for each block $B$ with respect to $c$ of $\cT$ 
that has~$c$ on the outer face. 
However, a priori we do not know at which block to root $\cT$.
Hence, our algorithm works as follows.

\paragraph{Algorithm.}
Given $G$, we can find its cut vertices and blocks and construct its block-cut tree $\cT$ in $\Oh(n)$ time~\cite{Tar72}.
Since we do not know what block can function as root of~$\cT$ yet, 
we start at the leaves of $\cT$ and work ``inwards''. 
Note that a leaf block $B$ has only one edge~$\set{B, c}$ in $\cT$ to a cut vertex $c$ (unless $B = G$) 
and we can thus provisionally direct~$\set{B, c}$ to $c$. 
This yields a block with respect to a cut vertex -- $B$ with respect to $c$.
Furthermore, during the algorithm, 
for at least one non-leaf block~$B$ either all or all but one of its neighboring blocks have been handled.
We then direct each edge~$\set{B, c'}$, where all other blocks adjacent to~$c'$ have been handled, towards~$B$.
If no undirected edge incident to~$B$ remains, then all neighboring blocks of~$B$ have been handled and~$B$ is the root.
Otherwise, there remains exactly one undirected edge~$\set{B, c}$.
Therefore, during this ad hoc post-order traversal of $\cT$, 
we can ensure that we always have at least one block~$B$ that is the root or for which we can provisionally direct 
the last undirected edge incident to~$B$ towards a cut vertex~$c$ in~$\cT$
i.e.,~$B$ becomes a block with respect to~$c$.

To process a block $B$ with respect to $c$, we check
whether $B$ has a 2-slope representation~$U_B$
with (i) $c$ on the outer face and (ii) additional constraints 
on the angles formed at all other cut vertices of $B$, which we describe in detail below.
If this is the case, then we can finalize the direction of the edge $\set{B, c}$ towards $c$.
Otherwise, if $B$ does not admit such a 2-slope representation, then~$B$ has to be the root of $\cT$.
We then orient all edges of $\cT$ towards $B$ and continue in the remaining part of $\cT$.
If we later find that another block also needs to be the root of $\cT$, 
then $G$ does not admit a 2-slope representation.
Furthermore, when we arrive back at $B$, 
we have to test whether~$B$ admits a 2-slope representation at all.

Note that for a block $B$ with respect to $c$, where $c$ has to be on the outer face of the 2-slope representation $U_B$,
the algorithms from the previous two sections only have to consider an SPQR-tree
with an edge incident to $c$ as reference edge.

\paragraph{Angles of cut vertices.}
We now describe the additional constraints 
that have to be checked for a block $B$ at all of $B$'s cut vertices.
More precisely, let $B$ be a block with respect to $c$ (or the root)
and let $c'$ be any other cut vertex of $B$ if it has any.
Depending on the degree of $c$ (and $c'$) in $B$ and in the neighboring blocks of $B$, 
we have the following extra conditions on the angles at $c$ and $c'$ in $U_B$ of $B$.

\begin{itemize}
\item Suppose $c$ (or $c'$) has degree one in $B$. Then $B$ is a single edge, 
	$c$ is automatically on the outer face in any 2-slope representation of $B$, 
	and the angle at $c$ in $U_B$ is insignificant.
\item Suppose $c$ has degree two in $B$ and $c$ has degree two in another block $B_1$; see \cref{fig:blockJoiningAngles}~(a). 
	Then $c$ has to have a large angle on the outer face in $U_B$,
	since otherwise it would not be able to attach to $U_{B_1}$ such that $B_1$ is in the outer face of~$B$. 
	If the same case applies to~$c'$, 
	then $c'$ has to have a large angle in $U_B$, but not necessarily on the outer face.
\item Suppose $c$ has degree two in $B$ and $c$ has degree one in two other blocks $B_1$ and~$B_2$.
	Then we first test if $B$ admits a 2-slope representation where $c$ has a large angle on the outer face; see \cref{fig:blockJoiningAngles}~(b).
	In this case, both $B_1$ and $B_2$ lie in the outer face of $B$. 
	Otherwise, if $c$ has indegree one and outdegree one in $B$, 
	we test whether~$B$ admits a 2-slope representation where $c$ forms a flat angle on the outer face; see \cref{fig:blockJoiningAngles}~(c).
	We test once such that $B_1$ lies on the outer face of $B$ and once for~$B_2$.
	In the former case, $B_2$ would lie in the interior of $B$ and thus $\set{B_2, c}$ would be directed as $(B_2, c)$;
	in the latter case, $B_1$ would lie in the interior of $B$ and thus $\set{B_1, c}$ would be directed as $(B_1, c)$.
	If neither is possible, then $B$ has to be the root. 
	For $c'$ there are no restrictions under these conditions.
\item The case where $c$ (or $c'$) has degree two in $B$ and degree one in exactly one other block is similar to the previous case but simpler.
\item Suppose $c$ has degree three in $B$; see \cref{fig:blockJoiningAngles}~(d). Then $c$ has to have a flat angle at $c$ on the outer face.
	Otherwise $B$ has to be the root.
	There are again no restrictions for~$c'$ under these conditions. 
\end{itemize}

\begin{figure}[htb]
  \centering
  \includegraphics{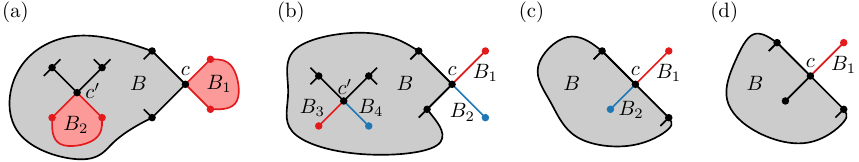}
  \caption{Conditions on the angles at $c$ and $c'$ in $U_B$ for block a $B$ with respect to $c$.}
  \label{fig:blockJoiningAngles}
\end{figure}

These conditions are clearly necessary and, following Healy and Lynch~\cite{HL07}, also sufficient.
Furthermore, they can easily be tested by the algorithms from the previous sections,
where we (as observed above) can use an edge $e$ incident to $c$ as reference edge. 
More precisely, if~$c$ has degree two in~$B$, then its two incident edges are merged in the root composition.
Hence, at this step we only allow a merge with the desired angle at~$c$,
that is, we check if there there is 2-slope representation of the child node of~$e$ with suitable upward spirality.
Otherwise, if~$c$ has degree three in~$B$ and has to form a flat angle,
we take the sole outgoing or sole incoming edge of~$c$ in~$B$ as the reference edge for the SPQR-tree.
In the root composition we then only allow a merge when there is the desired flat angle at~$c$ on the outer face.

\paragraph{Result.}
The total running time for our algorithm to test whether a general digraph admits a 2-slope representation
and thus an upward planar 2-slope drawing 
is given by (i) a linear amount for the computation of $\cT$, (ii) the sum of the checks for each block, 
and (iii) a linear amount for merging.
Because each cut vertex lies in at most four blocks,
the running time for (ii) is at most as much as if we tested a biconnected graph of the same size as $G$ once. 
Hence, we get the following results.

\begin{theorem}
Let $G$ be a series-parallel digraph with $n$ vertices.
There exists an $\Oh(n^4)$-time algorithm that tests if $G$ admits an upward planar 2-slope drawing and, if so,
that constructs such a drawing.
\end{theorem}

\begin{theorem} 
Let $G$ be a digraph with $n$ vertices. 
Let $t$ be the maximum number of nontrivial triconnected components of a block of $G$, and 
$d$ be the maximum diameter of a split component of a block of $G$. 
Then there exists an $\Oh(d^t t n^3 + dtn +d^2n^2)$-time algorithm that tests 
if~$G$ admits an upward planar 2-slope drawing and, if so,
that constructs such a drawing of~$G$.
\end{theorem}

\section{Phylogenetic networks} \label{sec:phynet} 
Recall from \cref{sec:introduction} that a phylogenetic network is a single-source digraph whose sinks are all leaves 
and whose non-sink, non-source vertices have degree three.
In this section we show how to find an upward planar 2-slope drawing of a phylogenetic network~$N$
such that its leaves lie on a horizontal line -- if~$N$ admits such a drawing.
Since we want that all leaves are on the outer face, we first merge them into a single vertex 
and then apply the linear-time algorithm of Bertolazzi \etal~\cite{BDMT98} to test whether the resulting digraph~$N'$ is upward planar. 
Clearly,~$N'$ is upward planar if and only if~$N$ admits a desired upward planar embedding.
In the affirmative case, let~$N$ now be an upward plane phylogenetic network such that its~$k$ leaves lie on the outer face.
Further assume that~$N$ contains no bad edge or, in this case equivalently, no transitive edge
(unlike the network in \cref{fig:phynet:bend}~(a)).

In \cref{sec:plane}, we constructed an upward planar 2-slope drawing by implementing the refinement step, 
which augments all faces to rectangular faces, and by applying a compaction algorithm~\cite{KKM01}. 
In order to obtain a drawing where all leaves lie on the same horizontal line,
we apply this algorithm to the following augmentation~$\bar{N}$ of~$N$.
Let~$l_1, l_2, \ldots, l_k$ be the leaves of~$N$ in clockwise order around the outer face.
Add new vertices~$v_1, v_2, \ldots, v_{k-1}$ 
and edges~$e_i = (l_i, v_i)$ and~$e'_{i}= (l_{i+1}, v_i)$, $i \in \set{1, \ldots, k-1}$; see \cref{fig:phynet}~(a).
Then apply \cref{clm:plane} to $\bar{N}$ to obtain an upward planar 2-slope drawing of~$\bar{N}$ in~$\Oh(n)$ time.

\begin{figure}[htb]
  \centering
  \includegraphics{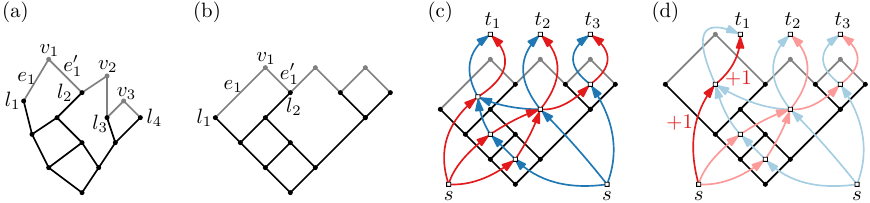}
  \caption{(a) The augmentation $\bar{N}$ of an upward planar phylogenetic network $N$;
  (b) an upward planar 2-slope drawing where $l_1$ and $l_2$ have different y-coordinates;
  (c) the dual flows $G_l$ (red) and $G_r$ (blue); 
  (d) Propagating the length difference of $e_1$ and $e_1'$ through $G_l$.}
  \label{fig:phynet}
\end{figure}

We observe from \cref{fig:phynet}~(b) that two vertices~$l_i$ and~$l_{i+1}$ of~$\bar{N}$ (neighboring leaves of~$N$)
have different y-coordinates if and only if~$e_i$ and~$e_i'$ have different lengths.
This can be fixed by propagating these length differences through the drawing in the following way (\cref{fig:phynet}~(c--d)).
Let~$G$ be the dual graph of~$\bar N$.
Furthermore, define~$G_l$ and~$G_r$ as the two subgraphs of~$G$ 
with~$V(G) = V(G_l) = V(G_r)$ 
and where~$E(G_l)$ and~$E(G_r)$ are the dual edges of primal edges with slope~$\sL$ and~$\sR$, respectively.
In other words,~$G_l$ is the dual graph of~$\bar N$ restricted to edges with slope~$\sL$. 
Direct every edge~$e^*$ in~$E(G_l)$ ($E(G_r)$) with primal edge~$e$ 
from the left (resp. right) face of~$e$ to the right (resp. left) face of~$e$.
Assign to each dual edge flow equal to the length of its primal edge.
Split the vertex corresponding to the outer face of~$\bar N$
into a source~$s$ and~$k-1$ sinks~$t_1, t_2, \ldots, t_{k-1}$ such that~$t_i$ has as incoming edges 
the dual edges of the primal edges~$e_i$ and~$e'_i$; see \cref{fig:phynet}~(c). 
These dual graphs can be constructed in linear time.

Next, to adjust the heights of the leaves of~$N$, 
for every pair~$e_i$ and~$e_i'$, $i \in \set{1, \ldots, k-1}$, if, say,~$e_i'$ is shorter
than~$e_i$, propagate the difference as flow backwards towards~$s$ through~$G_l$; see \cref{fig:phynet}~(d).
With one DFS on~$G_l$ and~$G_r$ each, all leaves can be handled simultaneously and in linear time.
Lastly, since some edge lengths have been changed, update the coordinates of all vertices in~$\bar{N}$ 
and remove the vertices~$v_i$, $i \in \set{ 1, \ldots, k}$, to obtain an upward planar 2-slope drawing of~$N$. 
The following theorem summarises this section. 

\begin{theorem} \label{clm:phynet}
Let $N$ be a phylogenetic network with $n$ vertices and no transitive edge.
If $N$ admits an upward planar drawing with all its leaves on the outer face,
then $N$ admits and upward planar 2-slope drawing such that all its leaves lie on a horizontal line.
Moreover, such a drawing can be constructed in $\Oh(n)$-time.
\end{theorem}

\begin{figure}[htb]
  \centering
  \includegraphics[page=2]{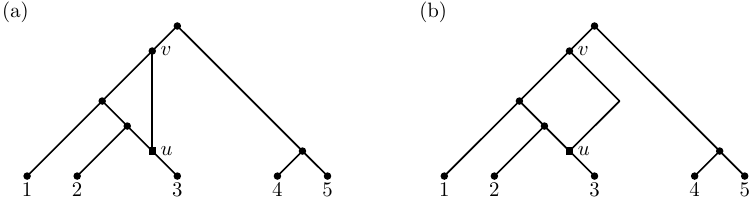}
  \caption{(a) A phylogenetic network with a transitive edge $(u, v)$ does not admit an upward planar 2-slope drawing,
  (b) however, it admits an upward planar 1-bend 2-slope drawing.}
  \label{fig:phynet:bend}
\end{figure}
 
Note that by \cref{clm:planar:singleSource:oneBend} a phylogenetic network with $m$ transitive edges 
admits a upward planar 1-bend 2-slope drawing with at most $m$ bends; see \cref{fig:phynet:bend}.

\section{Concluding remarks} 
When considering the number of slopes in a graph drawing, one typically asks how many different slopes are necessary for a graph of certain graph class.
Here we instead constrained the number of slopes to two and asked what digraphs can then be drawn upward planar.
Our digraphs are thus limited to those that contain no transitive edges and have a small maximum degree.  
Beyond that, the difficulty of the problem depends on whether or not an upward planar embedding is given 
and on the complexity of the digraph.

We have shown that if the embedding is fixed then the question can be answered and, 
in the affirmative, a drawing constructed in linear time. 
In this case the problem boils down to whether there is a bad edge for the given embedding and, if not, to adapt algorithms for orthogonal drawings.
However, even if there are bad edges present, allowing each of them to bend once is enough to obtain an upward planar 1-bend 2-slope drawing
with the minimum number of bends. 
We conjecture that it is NP-hard to minimize the drawing area of an upward planar 2-slope drawing just like it is for orthogonal drawings~\cite{Pat01}.
It would be interesting to see a proof for this and how compaction algorithms for orthogonal drawings can be applied to upward planar drawings. 

If a given digraph is not embedded yet, we first have to check whether the digraph is upward planar.
For single-source digraph, we have seen that it suffices to find one upward planar embedding, which may then be altered to one without bad edges if it exists.
For series-parallel and general digraphs we reused an approach by Didimo \etal~\cite{DGL10} based on SPQR-trees and upward spirality 
to find a quartic time and a fixed-parameter tractable algorithm, respectively.
An important difference is that our algorithm does not only compute upward planar embeddings for nodes of the SPQR-tree but also 2-slope representations.
Through the degree restrictions the algorithm became simpler and can thus also consider other properties.
It would be interesting to see whether the algorithm that computes an upward planar embedding of a single-source digraph can be modified
to directly compute a 2-slope representation. 

This research was motivated by drawings of phylogenetic networks.
While we here assumed that a given phylogenetic network is upward planar, this is not a biologically motivated property of phylogenetic networks.
One may argue that phylogenetic networks often have few reticulations (vertices with indegree two or higher), 
but even just two reticulations suffice to obstruct upward planarity.
Hence, it would be interesting to have algorithms that can also draw non-upward planar phylogenetic network with two slopes. 

The biggest challenges remain for drawings with more than two slopes.
Our feeling is that while the complexity of developing algorithms to draw graphs with two slopes is manageable,
three or more slopes increase the geometric interdependence dramatically.
While the companion paper by Klawitter and Zink~\cite{KZ21} started to investigate this, 
we would be happy to see more results on upward planar slope numbers of graphs.

\section*{Acknowledgements}
We thank the reviewers for their helpful comments and suggestions.

\pdfbookmark[1]{References}{References}
\bibliography{sources}
\bibliographystyle{abbrvurl}

\end{document}